\documentclass[11pt]{imsart}

\usepackage{graphicx}
\usepackage{amssymb}
\usepackage{amsmath}
\newtheorem{theorem}{Theorem}
\newtheorem{lemma}[theorem]{Lemma}

\newtheorem{corollary}[theorem]{Corollary}
\newtheorem{definition}[theorem]{Definition}
\newtheorem{remark}[theorem]{Remark}

\newenvironment{proof}[1][Proof]{\begin{trivlist}
\item[\hskip \labelsep {\bfseries #1}]}{\end{trivlist}}

\newcommand{\qed}{\nobreak \ifvmode \relax \else
      \ifdim\lastskip<1.75em \hskip-\lastskip
      \hskip1.75em plus0em minus0.75em \fi \nobreak
      \vrule height0.75em width0.50em depth0.25em\fi}

\startlocaldefs
\endlocaldefs

\graphicspath{ {images/} }

\begin{document}

\begin{frontmatter}

\title{A Predictive Model using the Markov Property}
\runtitle{SIFM}

\author{\fnms{Robert A.} \snm{Murphy, 
Ph.D.}\ead[label=e1]{robert.a.murphy@wustl.edu}}
\address{\printead{e1}}

\runauthor{Murphy}

\begin{abstract}
Given a data set of numerical values which are sampled from some unknown probability distribution, we will show how to check if the data set exhibits the Markov property and we will show how to use the Markov property to predict future values from the same distribution, with probability 1.
\end{abstract}



\end{frontmatter}


\section{The Problem}
\label{app}

\subsection{Problem Statement}

Given a data set consisting of numerical values which are sampled from some unknown probability distribution, we want to show how to easily check if the data set exhibits the Markov property, which is stated as a sequence of dependent observations from a distribution such that each successive observation only depends upon the most recent previous one.  In doing so, we will present a method for predicting bounds on future values from the same distribution, with probability 1.

\subsection{Markov Property}
\label{prop}

Let $I \subseteq \mathbb{R}$ be any subset of the real numbers and let $T \subseteq I$ consist of \textit{times} at which a numerical distribution of data is randomly sampled.  Denote the random samples by a sequence of random variables $\{X_{t}\}_{t \in T}$ taking values in $\mathbb{R}$.  Fix $t_{0} \in T$ and define $T_{0} = \{t \in T : t > t_{0}\}$ to be the subset of times in $T$ that are greater than $t_{0}$.  Let $t_{1} \in T_{0}$.

\begin{definition}
The sequence $\{X_{t}\}_{t \in T}$ is said to exhibit the \textbf{Markov Property}, if there exists a measureable function $Y_{t_{1}}$ such that
\begin{equation}
\label{XY}
X_{t_{1}} = Y_{t_{1}}(X_{t_{0}})
\end{equation}
\end{definition}
for all sequential times $t_{0},t_{1} \in T$ such that $t_{1} \in T_{0}$.

\subsection{Elementary Properties for a Specific Choice of Y}
\label{specific}

Let $T$ be defined as in section $(\ref{prop})$.  For $t \in T$, let $\epsilon_{t}$ define a sequence of independent, identically, normally distributed random variables, with $0$-mean and common, constant variance, $\sigma^{2}$.  Given $t_{1} \in T_{0}$, define a specific choice of the measureable function $Y_{t_{1}}$ to be such that
\begin{equation}
\label{YX}
Y_{t_{1}}(X_{t_{0}}) = X_{t_{0}} + \epsilon_{t_{1}}.
\end{equation}

Extending the setup in section $(\ref{prop})$, let $t_{2} < ... < t_{k} < t_{k+1} < ...$ be a sequence of times in $T$ with corresponding subsets $T_{1} \supseteq T_{2} \supseteq ... \supseteq T_{k} \supseteq T_{k+1} \supseteq ...$, where $t_{0} < t_{1} < t_{2}$ and $T_{0} \supseteq T_{1}$ such that
\begin{equation}
\label{eYX}
Y_{t_{k+1}}(X_{t_{k}}) = X_{t_{k}} + \epsilon_{t_{k+1}}.
\end{equation}

\begin{lemma}
\label{Z}
Fix finite $K \ge 1$ and let $T = \{t_{0},t_{1},...,t_{K}\}$ be restricted to being a finite set.  Then, for $k \in \{1,...,K\}$, it is true that
\begin{equation}
\label{X1}
X_{t_{k}} = X_{t_{0}} + Z_{k},
\end{equation}
\end{lemma}
where $Z_{k}$ is a normally distributed random variable, with $0$-mean and variance, $k\sigma^{2}$.

\begin{proof}
For each $k \in \{1,...,K\}$, successive application of eq. $(\ref{eYX})$ yields
\begin{equation}
\label{X2}
X_{t_{k}} = X_{t_{0}} + \sum_{j=1}^{k}\epsilon_{t_{j}}.
\end{equation}
Since $\epsilon_{t_{j}}$ are independent for all $j \in \{1,...,k\}$, then the random variable
\begin{equation}
\label{Z1}
Z_{k} = \sum_{j=1}^{k}\epsilon_{t_{j}}
\end{equation}
is easily shown to be normally distributed, with $0$-mean and variance, $k\sigma^{2}$, by the method of characteristics $\cite{Hogg}$.\qed
\end{proof}

\begin{remark}
\label{error}
By lemma $(\ref{Z})$, if we have a set of $K$ ordered, historical values, its error sequence $\{\epsilon_{k}\}_{k = 1}^{K}$ is estimated by $\epsilon_{k} = X_{t_{k}} - X_{t_{k-1}}$ for all $k \in \{1,...,K\}$.
\end{remark}

\begin{remark}
Also by lemma $(\ref{Z})$,  if the sampling times are re-labeled so that the last value in the ordering is re-labeled as $x_{0}$ corresponding to time $t_{0}$, then future values at times $k \ge 1$ are the sum of $x_{0}$ and $k$ independent, white noise (Gaussian) disturbances.  Thought of as a random walk, each new step in the \textit{``walk''} is simply a white noise disturbance.
\end{remark}

\begin{corollary}
\label{growth}
Let $\Omega$ be the sample space consisting of future values sampled from the probability distribution, $P$, such that $|\Omega| = H$.  Then, $P\bigg(|X_{t_{k}} - X_{t_{0}}| \le \sqrt{k}\sigma\bigg) = 1$ on $\Omega$, for all $k \in \{1,...,H\}$. 
\end{corollary}

\begin{proof}
Let $Var(X)$ denote the variance of the random variable, $X$.  Recalling that $\sigma^{2}$ is the variance of the historical error sequence and noting that the random variable $X_{t_{k}}-X_{t_{0}}$ is $0$-mean by lemma $(\ref{Z})$, then from eqs. $(\ref{X1})$, $(\ref{X2})$ and $(\ref{Z1})$, we have
\begin{eqnarray}
\int_{\Omega}(X_{t_{k}}-X_{t_{0}})^{2}dP &=& Var(X_{t_{k}}-X_{t_{0}}) \nonumber \\ &=& Var(Z_{k}) \nonumber \\ &=& k\sigma^{2}.
\end{eqnarray}
Since $(X_{t_{k}}-X_{t_{0}})^{2}$ is a non-negative random variable, then from Shiryaev $\cite{Shiryaev}$, we know that with probability 1,
\begin{equation}
(X_{t_{k}}-X_{t_{0}})^{2} \le \int_{\Omega}(X_{t_{k}}-X_{t_{0}})^{2}dP = k\sigma^{2}
\end{equation}
on $\Omega$.  Now, since $\{|X_{t_{k}}-X_{t_{0}}| \le \sqrt{k}\sigma\} = \{(X_{t_{k}}-X_{t_{0}})^{2} \le k\sigma^{2}\}$, then
\begin{equation}
P\bigg(|X_{t_{k}}-X_{t_{0}}| \le \sqrt{k}\sigma\bigg) = P\bigg((X_{t_{k}}-X_{t_{0}})^{2} \le k\sigma^{2}\bigg) = 1
\end{equation}
for all $k \in \{1,...,H\}$.\qed
\end{proof}

\begin{remark}
\label{shiryaev}
From Shiryaev $\cite{Shiryaev}$, if the chain $\{X_{t_{k}}\}_{k = 1}^{H}$ exhibits the Markov property, then it is independent of the start, $X_{t_{0}}=x_{0}$.
\end{remark}

\begin{remark}
If we run the Markov chain until a certain point, which we designate $X_{t_{0}} = x_{0}$, and sample $H$ times from the future, then by corollary $(\ref{growth})$, with probability $1$, we will not see growth beyond $x_{0} + \sqrt{H}\sigma$ nor will we see a decline below $x_{0} - \sqrt{H}\sigma$.
\end{remark}

\subsection{Check for the Markov Property}

The historical errors $\epsilon_{t_{k}}$ are assumed to be normally distributed for all $k \in \{1,...,K\}$.  Likewise, by the independence of each $\epsilon_{t_{k}}$ from each $\epsilon_{t_{j}}$, for all $k,j \in \{1,...,K\}$, when $k \ne j$, then the variance of the sum of the errors is just $K\sigma^{2}$.  Therefore, by remark $(\ref{error})$, we only need to show that $X_{t_{k}}-X_{t_{k-1}}$ is normally distributed, with $0$-mean and constant variance $\sigma^{2}$, for all $k \in \{1,...,K\}$, in order to show that the sequence of data measurements $\{X_{t_{k}}\}_{k=1}^{K}$ is Markovian, with respect to the chosen random model, $Y_{t_{k}}$, given in section $(\ref{specific})$.

Define a test statistic $W$ as
\begin{equation}
W = \frac{\bigg(\sum_{k=1}^{K}a_{k}x_{(k)}\bigg)^{2}}{\sum_{i=1}^{K}(x_{i}-\overline{x})^{2}},
\end{equation}
where $x_{(k)}$ and $\overline{x}$ are the $k^{th}$ element in an ordering of $\{x_{k}\}_{k = 1}^{K}$ and its sample mean, respectively, and $(a_{1},...,a_{K})$ is computed as
\begin{equation}
(a_{1},...,a_{K}) = \frac{m^{T}V^{-1}}{m^{T}V^{-1}V^{-1}m},
\end{equation}
such that $m = (m_{1},...,m_{K})$ is a vector of expected values of the order statistics used to give the ordering $\{x_{(k)}\}_{k=1}^{K}$ and $V$ is the covariance matrix of the order statistics.

\begin{definition}
\label{SW}
The \textbf{Shapiro-Wilk Test of Normality} is the test statistic $W$, such that, if a level of significance ($p$-value) is assigned in a hypothesis test, where the null hypothesis is that the sample was drawn from a normal distribution, then a value of $W$ which exceeds the probability $(1-2p)$ affirms the null hypothesis.
\end{definition}

\noindent
The Shapiro-Wilk test now provides a sufficient condition for testing if the sequence of errors $\{\epsilon_{t_{k}}\}_{k = 1}^{K}$, defined in our Markov model $\{Y_{t_{k}}\}_{k = 1}^{K}$, is normally distributed, which amounts to $\{X_{t_{k}}\}_{k=1}^{K}$ forming a Markov chain with respect to the model $\{Y_{t_{k}}\}_{k = 1}^{K}$.

\section{Airline Schedule Interruption Counts Exhibit the Markov Property}

\subsection{Problem Statement}

To a manufacturer of large airliners, a schedule interruption is any event that causes an airliner to be more than 15 minutes late on its scheduled departure time from an airport or more than 15 minutes late arriving into an airport due to mechanical or electrical failure of a part, subsystem or system on said aircraft.  Given a data set containing a $K$-month period of historical schedule interruption counts, we will present a calculation of bounds on the number of schedule interruptions in the following $H$-month future period, after which, we want to be able to calculate bounds on the total cost impact.

\subsection{Bounds on Schedule Interruption Counts}
\label{bounds}

Using def. $(\ref{SW})$, the errors obtained from a proprietary schedule interruptions data set gives a value of $W \approx 0.90$, which is right at the level of significance when we set $p = 0.05$.  We accept the null hypothesis and conclude that the sequence of errors was drawn from a normal distribution so that the original data set is Markovian, according to the model of $Y_{t_{k}}$ given by eq. $(\ref{YX})$.  Therefore, by remark $(\ref{shiryaev})$, we can run the chain up to the end and label this point $x_{0}$.  Then, with probability $1$, future schedule interruption counts will not increase beyond $x_{0} + \sqrt{H}\sigma$ nor decrease below $x_{0} - \sqrt{H}\sigma$, where $H$ is the number of future data points and $\sigma^{2}$ is the variance of the past data points, up to $x_{0}$.

\subsection{Model of Cost Impact Due to Schedule Interruptions}

Using lemma $(\ref{X1})$ and eq. $(\ref{X2})$ in section $(\ref{app})$, we see that the future interruption counts are the sum of the last interruption count plus $0$-mean, white noise with variance, $H\sigma^{2}$, obtained in section $(\ref{bounds})$.  By the Markov property, also shown in section $(\ref{bounds})$, noise associated with future schedule interruption counts is $0$-mean with respect to $x_{0}$, so that we have a normal distribution with mean given by a horizontal line extending from $x_{0}$, of length exactly $H$ months.  To complete the future data space, we have the familiar ``bell'' shape with standard deviation $\sqrt{k}\sigma$, for $k \in \{1,2,...,H\}$ corresponding to each month in the $H$-month future time span, which extends beyond the end of the historical data set.

\subsubsection{Average Monthly Cost Per Schedule Interruption}
\label{avgcost}

With an application of the Central Limit Theorem, we can make the assumption of an approximately normal distribution for the total cost impact due only to delays $(D)$, cancellations $(C)$, diversions $(d)$ and air-turn-backs $(A)$.  Hence, using a maximum likelihood technique, we see that the best estimate of the true mean of the distribution is given by the average cost impact for the sum of the historical counts of the different delay classes. As such, we first calculate the \textbf{delay class average monthly cost impact (ADC)} per interruption for a $K$-month historical period as

\begin{equation}
ADC = \frac{\sum_{k=1}^{K}\bigg((C_{D}D_{k} + C_{C}C_{k} + C_{d}d_{k} + C_{A}A_{k})/(D_{k}+C_{k}+d_{k}+A_{k})\bigg)}{K},
\end{equation}
where $C_{D}$, $C_{C}$, $C_{d}$ and $C_{A}$ are the average costs associated with delays, cancellations, diversions and air-turn-backs.

Making another appeal to the Central Limit Theorem for the distribution of total cost impact due only to spares replacements $(S)$, we next calculate an estimate of the mean of the distribution as

\begin{equation}
ASC = \frac{\sum_{k=1}^{K}\bigg(C_{S}S_{k}/(D_{k}+C_{k}+d_{k}+A_{k})\bigg)}{K},
\end{equation}
where $ASC$ is the \textbf{average monthly spares cost} per interruption.  Thus, our average monthly costs per schedule interruption in the $K$-month period is estimated to be the sum $ADC+ASC$.  Now, by corollary $(\ref{growth})$, for each month in the $H$-month period beyond the end of the historical data set, our \textbf{future average cost impact} due to schedule interruptions is bounded below by $(x_{0}-\sqrt{k}\sigma)*(ADC+ASC)$ and bounded above by $(x_{0}+\sqrt{k}\sigma)*(ADC+ASC)$ for each $k \in \{1,2,...,H\}$.  By corollary $(\ref{growth})$, with probability 1, these costs bound our future, $H$-month total cost impact to provide the familiar ``bell'' shape of our normally distributed future data set, by lemma $(\ref{Z})$.

\subsubsection{Cost Prediction of Schedule Interruption Counts}

Now that we have our average schedule interruption costs bounded with probability 1, as stated in section $(\ref{avgcost})$, a prediction of future average costs per schedule interruption can be obtained by sampling from the normal distribution whose mean and variance are given by $x_{0}*(ADC+ASC)$ and $k\sigma^{2}*(ADC+ASC)$, respectively, for each $k \in \{1,...,H\}$.

\end{document}